\documentclass{article}
\usepackage{amsmath}
\usepackage{amsthm}
\usepackage{graphicx}
\usepackage[utf8]{inputenc}
\usepackage{pmboxdraw}
\newcommand{\step}[1]{\stackrel{#1}{\longrightarrow}}
\newcommand{\bisim}[2]{#1\mathbin{\underline{\leftrightarrow}}#2}
\newcommand{\teq}[2]{\mathrm{teq}\,(#1,#2)}
\newcommand{\congr}[2]{\mathrm{congr}\,(#1,#2)}
\newcommand{\nf}[1]{\mathrm{nf}\,(#1)}
\newcommand{\nfmult}[2]{\mathrm{nfmult}\,(#1,#2)}
\newcommand{\tail}[2]{\mathrm{tail}\,(#1,#2)}
\newcommand{\nextt}[1]{\mathrm{next}\,(#1)}
\newcommand{\obl}[3]{\mathrm{obl}\,(#1,#2,#3)}
\newcommand{\mult}[2]{\mathrm{mult}\,(#1,#2)}
\newcommand{\binstar}[2]{#1\mathbin{{}^{\scalebox{1.1}{$*$}}}\hspace{-2pt}#2}
\newtheorem{conjecture}{Conjecture}
\newtheorem{theorem}{Theorem}
\newtheorem{lemma}{Lemma}
\title{A Complete Axiom System for 1-Free Kleene Star
  Expressions under Bisimilarity: An Elementary Proof}
\author{Allan van Hulst}
\date{\texttt{allanvanhulst@protonmail.com}}
\begin{document}
\maketitle
\begin{abstract}
Grabmayer and Fokkink recently presented a finite and complete axiomatization
for 1-free process terms over the binary Kleene star under bismilarity equivalence 
(proceedings of LICS 2020, preprint available). A different and considerably simpler 
proof is detailed in this paper. This result, albeit still somewhat technical, only 
relies on induction and normal forms and is therefore also much closer to a potential 
rewriting algorithm.  In addition, a complete verification in the Coq proof assistant 
of all results in this work is provided, but correctness does not depend upon any 
computer-assisted methodology.
\end{abstract}
\section{Introduction}
\label{sec:introduction}
A completeness proof for a 1-free process theory modulo bisimilarity was recently 
presented in \cite{gf20}\footnote{Preprint available at: 
\texttt{https://www.cs.vu.nl/{\textasciitilde}wanf/publications.html}}\footnote{Extended 
version available at: \texttt{https://arxiv.org/abs/2004.12740}}. 
Being somewhat daunted by the complexity of this proof provided a powerful incentive 
to search for a simpler solution, which is presented in this work. This different and 
considerably simpler proof only uses induction and normal forms, and has been additionally 
verified by means of the Coq proof assistant.

This paper is an intermediate step in a research line towards the resolution of a question 
originally posed in \cite{mil84}: does there exist a finite and complete
axiomatization for the unary Kleene star under bisimilarity equivalence? This problem
is usually considered in the context of process theories including the constants 0
(deadlock), 1 (empty process) and the operators + (non-deterministic choice) and
$\cdot$ (sequential composition). Milner himself suggested that solving this problem
may involve a considerable effort \cite{mil84}. It is clear that the question remains 
unanswered to this day.

Earlier attempts in this direction include a completeness proof in absence of both
the constants 0 and 1 in \cite{bbp94}, and a variant where every Kleene star appears 
as $p^*\cdot 0$ in \cite{fok97}. Completeness proofs for simpler theories (e.g. 
without the Kleene star) can be found in any process algebra handbook (cf. \cite{jb10}).

The new approach in \cite{gf20} deviates from the rather syntactic treatment in earlier 
works and instead takes the more semantic avenue of using process charts. This toolset is 
applied to prove completeness for a 1-free process theory over the binary Kleene star 
modulo bisimilarity. While novel and innovative, this results in a quite complex proof.

This paper presents a two-fold approach. First, it is shown that every process term
has a bisimulant in a normal form and of non-increased star nesting depth. In essence,
our normalization requirement conforms to the following congruence property:
if term $p$ reduces to both $t$ and $u$ in one or more steps then bisimilarity of
$t$ and $u$ is a consequence of bisimilarity of $t\cdot\binstar{p}{q}$ and $u\cdot
\binstar{p}{q}$. The second part of the proof applies induction towards the star
nesting depth and proves equality under the condition that one operand is normalized,
which is sufficient due to symmetry and transitivity. 

This paper is set up in self-contained form. No claim is made that the supporting Coq 
code is the most neat or elegant reflection of such a proof in formalized mathematics, 
it just serves as an additional layer of verification.

The remainder of this paper is organized as follows. Section \ref{sec:definitions}
contains a number of basic definitions and section \ref{sec:soundness} concerns 
the soundness of the axiom system. These basic preliminaries are then followed
by three sections in which the the completeness proof is built from the ground up.
Section \ref{sec:summation} defines a summation operator and proves some basic
results. Normal forms of terms are defined and shown to be derivable under 
bisimilarity in section \ref{sec:normalization}. These results are then integrated
into a completeness proof in section \ref{sec:completeness}. The formalization of
the theory in the Coq proof assistant is detailed in section \ref{sec:formalization},
which mainly serves to make this material more accessible to readers who are less
familiar with such techniques. A short concluding section \ref{sec:conclusions}
refines Milner's completeness problem to a more detailed conjecture.
\section{Definitions}
\label{sec:definitions}
We will mostly follow standard nomenclature and definitions in process algebra and
propose the books \cite{jb10} and \cite{fok13} as general reference texts. Throughout 
this work, we will assume $A$ to be a set of \emph{actions}. There is no requirement that 
$A$ is finite, as this proof concerns the completeness of closed terms only. Actions form 
the elementary operations in the set of process terms $T$ defined inductively as
\begin{center}
\begin{math}
T\mathbin{::=}\,\,0\mid A\mid T+T\mid T\cdot T\mid\binstar{T}{T}.
\end{math}
\end{center}
In the process algebra $T$ the constant $0$ expresses deadlock (i.e.
a process exhibiting no behavior). Every action $a\in A$ induces a step 
$a$ to the special termination symbol $\surd$ as its sole behavior (note
that $\surd$ is not part of the algebra $T$). The operation $p+q$ 
models a (possibly non-deterministic) choice between $p$ and $q$, while the 
sequential composition $p\cdot q$ denotes the concatenation of the behaviors
of $p$ and $q$. The binary Kleene star $\binstar{p}{q}$ models zero or more 
iterations of $p$, possibly followed by $q$. Sequential composition is assumed 
to be right-associative, while the other operators associate to the left. The 
Kleene star binds stronger than sequential composition, which in turn binds 
stronger than plus.

Process terms exhibit behavior defined as taking actions resulting in
either a new process term or the termination symbol $\surd$. We set 
$V=T\cup\{\surd\}$ and define a relation $\longrightarrow
\subseteq V\times A\times V$ to formally capture this behavior. Assume that 
$p,q,p',q'\in T$, $v\in V$, and $a\in A$ in the set of derivation rules listed 
below, using the notation $p\step{a}q$ for $(p,a,q)\in\longrightarrow$.
\begin{center}
\begin{tabular}{c}
$\displaystyle\frac{}{a\step{a}\surd}$
\qquad
$\displaystyle\frac{p\step{a}v}{p+q\step{a}v}$
\qquad
$\displaystyle\frac{q\step{a}v}{p+q\step{a}v}$
\qquad
$\displaystyle\frac{p\step{a}p'}{p\cdot q\step{a}p'\cdot q}$ 
\bigskip \\
$\displaystyle\frac{p\step{a}\surd}{p\cdot q\step{a}q}$
\qquad
$\displaystyle\frac{p\step{a}p'}{\binstar{p}{q}\step{a}p'\cdot\binstar{p}{q}}$
\qquad
$\displaystyle\frac{p\step{a}\surd}{\binstar{p}{q}\step{a}\binstar{p}{q}}$
\qquad
$\displaystyle\frac{q\step{a}v}{\binstar{p}{q}\step{a}v}$ 
\bigskip \\
\end{tabular}
\end{center}
Bisimilarity is a coinductively defined relationship which relates process 
terms in a lock-step fashion. The proof in \cite{gf20} defines bisimilarity
in terms of process charts. As these constructs are not used in this
proof, we will employ a definition in terms of $V$. 

Elements $u,v\in V$ are bisimilar (notation $\bisim{u}{v}$) if 
there exists a relation $R\subseteq V\times V$ such that 
$(u,v)\in R$ and for all $(x,y)\in R$ the following are satisfied
\begin{enumerate}
\item $x = \surd$ if and only if $y = \surd$,
\item for all $x\step{a}x'$ there exists a $y'\in V$ such that
      $y\step{a}y'$ and $(x',y')\in R$, and
\item for all $y\step{a}y'$ there exists an $x'\in V$ such that
      $x\step{a}x'$ and $(x',y')\in R$.
\end{enumerate}
Examples of pairs of bisimilar terms include $\bisim{\binstar{(\binstar{a}{b})}{c}}{\binstar{(a+b)}{c}}$ 
and $\bisim{\binstar{(a+aa)}{0}}{\binstar{a}{0}}$, but the terms
$a\cdot b+a\cdot c$ and $a\cdot(b+c)$ are not bisimilar.

The following set of axioms will be shown to be sound and complete
with regard to bisimilarity in this paper. In a very strict context,
these should be interpreted as axiom-schemes, in the sense that for 
each closed instance of the variables, the corresponding axiom is 
defined.
\begin{center}
\begin{tabular}{lrcl|llrcl}
(B1) & $x+y$ & = & $y+x$ & & (B6) & $x+0$ & = & $x$ \\
(B2) & $(x+y)+z$ & = & $x+(y+z)$ & & (B7) & $0\cdot x$ & = & $0$ \\
(B3) & $x+x$ & = & $x$ & & (BKS1) & $x\cdot\binstar{x}{y}+y$ & = & $\binstar{x}{y}$ \\
(B4) & $(x+y)\cdot z$ & = & $x\cdot z+y\cdot z$ &
 & (BKS2) & $(\binstar{x}{y})\cdot z$ & = & $\binstar{x}{(y\cdot z)}$ \\
(B5) & $(x\cdot y)\cdot z$ & = & $x\cdot(y\cdot z)$ & 
& (RSP) & $x$ & = & $y\cdot x+z$\,\,\,\,\textrm{implies} \\
 & & & & & & $x$ & = & $\binstar{y}{z}$ \\
\end{tabular}
\end{center}
The axiomatization is a straightforward adaptation of the set of 
axioms originally proposed in \cite{mil84}, which was in turn 
adapted from a language-theoretic setting based on work by Salomaa
\cite{sal66}). It is well-known that at least one higher-order 
construct is required as shown in \cite{sew94}. Variants of
the axiom of the recursive specification principle (RSP) have been
studied extensively (cf. \cite{jb10}).

For the remainder of this paper, the notation $p=q$ will be used to
denote axiomatic equality whereas $p\equiv q$ will be used to denote
exact syntactic equality (e.g. $a+(b+c)=(a+b)+c$ but 
$a+(b+c)\not\equiv(a+b)+c$).

The definitions listed here are sufficient to formulate the main 
result of this work, the proof of which is divided over several succeeding
sections.
\begin{theorem}
For all $p,q\in T$ such that $\bisim{p}{q}$ it holds that $p=q$.
\end{theorem}
\section{Soundness}
\label{sec:soundness}
We briefly consider soundness of the theory to fulfill the objective 
of being self-contained and to state a simple lemma which serves as a 
very useful building block in the remainder of the proof. 
\begin{lemma}
\label{lem:bisim_next}
For all $p,q\in T$ such that
\begin{enumerate}
\item for all $p\step{a}p'$ there exists a $q\step{a}q'$ such that $\bisim{p'}{q'}$ and
\item for all $q\step{a}q'$ there exists a $p\step{a}p'$ such that $\bisim{p'}{q'}$,
\end{enumerate}
it holds that $\bisim{p}{q}$.
\end{lemma}
Soundness of the axiom RSP is neither deep nor straightforward.
\begin{lemma}
\label{lem:rsp_sound}
If $\bisim{p}{q\cdot p+r}$ then $\bisim{p}{\binstar{q}{r}}$, for all $p,q,r\in T$.
\end{lemma}
\begin{proof}
Assume $R$ is the witnessing relation for $\bisim{p}{q\cdot p+r}$. It is
straightforward to show that the transitive closure $\overline{R}$
of $R$ is again a bisimulation. If $R'$ is defined as:
\begin{center}
\begin{math}
R'=\{(p,\binstar{q}{r})\}\cup\overline{R}\cup\{(p',q'\cdot\binstar{q}{r})
  \mid(p',q'\cdot p)\in\overline{R}\}\cup\{(p',\binstar{q}{r})\mid(p',p)\in\overline{R}\}
\end{math}
\end{center}
then $R'$ can be chosen as a witnessing relation for $\bisim{p}{\binstar{q}{r}}$.
\end{proof}
Soundness of the theory is required as a lemma in the completeness proof.
\begin{lemma}
\label{lem:soundness}
For all $p,q\in T$ such that $p=q$ it holds that $\bisim{p}{q}$.
\end{lemma}
\begin{proof}
Assume that $p=q$ and apply induction towards the derivation tree of $p=q$. 
Most of the cases can be resolved easily via Lemma \ref{lem:bisim_next}.
Soundness for the axioms B5 and BKS2 is only slightly more complicated.
Soundness for the axiom RSP is shown in Lemma \ref{lem:rsp_sound}.
\end{proof}
\section{Summation}
\label{sec:summation}
Expressing process terms as sums forms a crucial step between algebraic
and semantic reasoning in the succeeding proofs. For example, the proof
of Lemma \ref{lem:next} becomes much easier once we are able to obtain
such sums under equality. 

For finite sets $N\subseteq A \times V$ we define the summation operator 
$\sigma:\{N\subseteq A\times V\}\rightarrow T$ recursively by setting 
$\sigma(\emptyset)=0$ and
\begin{center}
\begin{math}
\sigma(\{(a,u)\}\cup N)=\sigma(N)+\begin{cases}
                                  a & \textrm{if}\,\,u\,\,\textrm{equals}\,\,\surd \\
                                  a\cdot u & \textrm{if}\,\,u\in T \\
                                  \end{cases}
\end{math}
\end{center}
where the expression $\{(a,u)\}\cup N$ refers to the disjoint union
(i.e. $(a,u)\not\in N$). 
\begin{lemma}
\label{lem:summation}
For all $p\in T$ there exists an $N\subseteq A\times V$
such that $p=\sigma(N)$ and for all $(a,u)\in A\times V$ it holds
that $p\step{a}u$ if and only if $(a,u)\in N$.
\end{lemma}
\begin{proof}
Apply induction towards the structure of $p$. In case $p\equiv 0$, choose $\emptyset$
for $N$. If $p\equiv a$ for some $a\in A$, choose $\{(a,\surd)\}$ for $N$. For the case 
$p\equiv p_1+p_2$, observe that in general $\sigma(N_1\cup N_2)=\sigma(N_1)+\sigma(N_2)$
always holds.  

If $p\equiv p_1\cdot p_2$, first use induction to obtain $N_1$ such that 
$p_1=\sigma(N_1)$. Then, apply the following projection to each element
$(a,u)\in N_1$: (1) if $u$ equals $\surd$ then project to $(a,p_2)$, (2) 
if $u\in T$ then project to $(a,u\cdot p_2)$ and name the result $N'$. 
Subsequently, we have $p_1\cdot p_2=\sigma(N')$ using B5. 

For the case $p\equiv\binstar{p_1}{p_2}$, use BKS1 to rewrite as $p_1\cdot
\binstar{p_1}{p_2}+p_2$ and treat the part $p_1\cdot\binstar{p_1}{p_2}$ in
the same way as for $p\equiv p_1\cdot p_2$.
\end{proof}
We apply Lemma \ref{lem:summation} to obtain a useful intermediate result 
stated in Lemma \ref{lem:next}, which is very similar in form to Lemma 
\ref{lem:bisim_next}. We define the predicate $\teq{u}{v}$ for all $u,v\in 
V$ as shown below, where $\mathrm{teq}$ stands for termination-or-(axiomatically)equal.
\begin{center}
\begin{math}
\teq{\surd}{\surd}=\mathit{true}\qquad
\teq{p}{q}\iff p=q\,\,\textrm{if}\,\,p,q\in T
\end{math}
\end{center}
and define $\teq{u}{v}$ as $\mathit{false}$ in all other cases.
\begin{lemma}
\label{lem:next}
For all $p,q\in T$ such that
\begin{enumerate}
\item for all $p\step{a}p'$ there exists a $q\step{a}q'$ such that $\teq{p'}{q'}$ and
\item for all $q\step{a}q'$ there exists a $p\step{a}p'$ such that $\teq{p'}{q'}$,
\end{enumerate}
it holds that $p=q$.
\end{lemma}
\begin{proof}
Assume $p=\sigma(M)$ and $q=\sigma(N)$ for $M,N\subseteq A\times V$
as derived by Lemma \ref{lem:summation}. Rewrite as $p=\sigma(M)=\sigma(M)+\sigma(N)=
\sigma(N)=q$ and solve the two equalities $\sigma(M)=\sigma(M)+\sigma(N)$ and 
$\sigma(N)=\sigma(M)+\sigma(N)$ separately by induction towards the size of the
set that appears once in the respective equality.
\end{proof}
\section{Normalization}
\label{sec:normalization}
In the remainder, let $p\longrightarrow^*q$ denote the fact that $p$ reduces to
$q$ in zero or more steps for $q\in T$. Similarly, we define $p\longrightarrow^+q$
representing a reduction in one or more steps for $q\in T$. 

The core of this completeness proof uses the result that terms in $T$ can
be normalized under bisimilarity such that for every subterm $\binstar{r}{s}$ and
all reductions $r\longrightarrow^+p$ and $r\longrightarrow^+q$ such that 
$\bisim{p\cdot\binstar{r}{s}}{q\cdot\binstar{r}{s}}$ it holds that $\bisim{p}{q}$.
The precise meaning of 'subterm' will become more clear in a short while. We require 
a predicate $\mathrm{congr}$ to express (a premise for) this congruence property:
\begin{center}
\begin{math}
\congr{p}{q}\iff\textrm{for all}\,\,p\longrightarrow^+t:\,\,\bisim{t\cdot q}{q}
  \,\,\textrm{does not hold}.
\end{math}
\end{center}
Informally, if there exists a $p\longrightarrow^+t$ such that 
$\bisim{t\cdot q}{q}$ then $\congr{p}{q}$ is false, and true if no such 
$p\longrightarrow^+t$ exists. We may now prove a crucial lemma.
\begin{lemma}
\label{lem:congruence}
If $\bisim{p\cdot r}{q\cdot r}$ and $\congr{p}{r}$ and $\congr{q}{r}$
then $\bisim{p}{q}$, for all $p,q,r\in T$.
\end{lemma}
\begin{proof}
Assume $R$ is a witnessing relation for $\bisim{p\cdot r}{q\cdot r}$
and define $R'$ as follows:
\begin{center}
\begin{math}
R'=\{(\surd,\surd)\}\cup\{(p,q)\}\cup\{(p',q')\mid (p'\cdot r,q'\cdot r)\in R,\,
  p\longrightarrow^+p',\,q\longrightarrow^+q'\}
\end{math}
\end{center}
Assume there exists $(p'\cdot r,q'\cdot r)\in R$ such that $(p',q')\in R'$ for some
$p',q'\in T$. If there exists a step $p'\cdot r\step{a}p''\cdot r$ for $p''\in T$
and a step $q'\cdot r\step{a}r$ (i.e. when $q'\step{a}\surd$) such that
$\bisim{p''\cdot r}{r}$ then this contradicts $\congr{p}{r}$. The argument is 
symmetric.
\end{proof}
The congruence property may be used to recursively define a normal form,
thereby making the notion of subterm with regard to the congruence property 
more precise. Two remarks are important now.

First, suppose we have a term $(\binstar{p}{q})\cdot(\binstar{r}{s})$ then 
two properties are desired:
\begin{enumerate}
\item for $p\longrightarrow^+t$ and $p\longrightarrow^+u$ we have:
      $\bisim{t\cdot\binstar{p}{q}\cdot\binstar{r}{s}}{u\cdot\binstar{p}{q}\cdot\binstar{r}{s}}$
      implies $\bisim{t}{u}$ and
\item for $r\longrightarrow^+x$ and $r\longrightarrow^+y$ we have:
      $\bisim{x\cdot\binstar{r}{s}}{y\cdot\binstar{r}{s}}$ implies $\bisim{x}{y}$.
\end{enumerate}
Therefore, the cases for $\binstar{p}{q}$ and $\binstar{r}{s}$ cannot be
separated into two conditions relying solely on $\binstar{p}{q}$ and
$\binstar{r}{s}$. This is resolved by using a binary predicate to express
the fact that a term is normalized.

In general, the process algebra $T$ does not have a neutral element under sequential 
composition. This necessitates the definition of two slightly different
predicates for normal forms. Although this makes the proof a slightly more
difficult, this does not present a fundamental complication.

We define two normal form predicates $\mathrm{nfmult}$ and $\mathrm{nf}$ recursively
as shown below. Note that there is no mutual dependence between $\mathrm{nf}$ and
$\mathrm{nfmult}$.
\begin{center}
\begin{tabular}{rcl}
$\nfmult{0}{q}$ & $\iff$ & $\mathit{true}$ \\
$\nfmult{a}{q}$ & $\iff$ & $\mathit{true}$ \\
$\nfmult{p_1+p_2}{q}$ & $\iff$ & $\nfmult{p_1}{q}\,\,\textrm{and}\,\,\nfmult{p_2}{q}$ \\
$\nfmult{p_1\cdot p_2}{q}$ & $\iff$ & $\nfmult{p_1}{p_2\cdot q}\,\,\textrm{and}\,\,\nfmult{p_2}{q}$ \\
$\nfmult{\binstar{p_1}{p_2}}{q}$ & $\iff$ & $\nfmult{p_1}{\binstar{p_1}{p_2}\cdot q}\,\,\textrm{and}\,\,
  \nfmult{p_2}{q}\,\,\textrm{and}$ \\
& & $\congr{p_1}{\binstar{p_1}{p_2}\cdot q}$ \\
\end{tabular}
\end{center}
and
\begin{center}
\begin{tabular}{rcl}
$\nf{0}$ & $\iff$ & $\mathit{true}$ \\
$\nf{a}$ & $\iff$ & $\mathit{true}$ \\
$\nf{p_1+p_2}$ & $\iff$ & $\nf{p_1}\,\,\textrm{and}\,\,\nf{p_2}$ \\
$\nf{p_1\cdot p_2}$ & $\iff$ & $\nfmult{p_1}{p_2}\,\,\textrm{and}\,\,\nf{p_2}$ \\
$\nf{\binstar{p_1}{p_2}}$ & $\iff$ & $\nfmult{p_1}{\binstar{p_1}{p_2}}\,\,\textrm{and}\,\,
  \nf{p_2}\,\,\textrm{and}$ \\
& & $\congr{p_1}{\binstar{p_1}{p_2}}$ \\
\end{tabular}
\end{center}
We consider several simple examples. Observe that $\nf{\binstar{(a\cdot b+a)}{0}}$ 
holds because there does not exist a reduction sequence $a\cdot b+a\longrightarrow^+t$
such that $\bisim{t\cdot\binstar{(a\cdot b+a)}{0}}{\binstar{(a\cdot b+a)}{0}}$. For
the term $\binstar{(a\cdot\binstar{a}{a})}{0}$, such a reduct $a\cdot\binstar{a}{a}
\longrightarrow^+\binstar{a}{a}$ indeed exists. However, the bisimulant 
$\binstar{(a+a)}{0}$ of $\binstar{(a\cdot\binstar{a}{a})}{0}$ is normalized
and constructed as such in Lemma \ref{lem:congr_mult} and Lemma \ref{lem:congr_ex}.

The term $\binstar{(a\cdot\binstar{(a\cdot(b\cdot a+a))}{c})}{0}$ is an example originally 
proposed in \cite{fok97} that is re-used in the recent result of Fokkink and Grabmayer \cite{gf20}. For
compactness we abbreviate these as $\binstar{q}{0}\equiv\binstar{(a\cdot\binstar{(a\cdot(b\cdot a+a))}{c})}{0}$ 
and $p\equiv a\cdot(b\cdot a+a)$ such that $\binstar{q}{0}\equiv\binstar{(a\cdot\binstar{p}{c})}{0}$.
Now observe that $p\longrightarrow^+a$ such that 
$\bisim{(a\cdot\binstar{p}{c})\cdot\binstar{q}{0}}{\binstar{q}{0}}$. In
this case, we can obtain a normalized bisimulant
$\bisim{\binstar{q}{0}}{a\cdot\binstar{(p+c\cdot a)}{0}}$. The setup of Lemma 
\ref{lem:congr_mult} and Lemma \ref{lem:congr_ex} is precisely tailored to 
construct normal forms for these types of cases.

We now prove two straightforward results concerning the normal form 
predicates.

\begin{lemma}
\label{lem:right_compat}
Both $\mathrm{nfmult}$ and $\mathrm{congr}$ are right-compatible
under bisimilarity.
\end{lemma}
\begin{proof}
For the first result, use induction towards $p$ to prove that 
$\nfmult{p}{r}$ is a consequence of $\nfmult{p}{q}$, given $\bisim{q}{r}$.
A similar result for $\mathrm{congr}$ follows directly from the definition.
\end{proof}

\begin{lemma}
\label{lem:preserved}
Both $\mathrm{nf}$ and $\mathrm{nfmult}$ are preserved under $\longrightarrow$.
\end{lemma}
\begin{proof}
Use induction towards $p$ to to prove that $\nf{q}$ is a consequence of
$\nf{p}$ and $p\step{a}q$, for some $a\in A$. Similarly, induction towards
$p$ can be applied to derive $\nfmult{q}{r}$ from $\nfmult{p}{r}$ if
$p\step{a}q$ for some $a\in A$ and $r\in T$.
\end{proof}

We are now ready to prove the two key lemmas for deriving a bisimilar term 
satisfying the congruence property. Note that Lemma \ref{lem:congr_mult} is
the first point in the proof where we will use induction towards the star-depth
$d$ which is defined straightforwardly as shown below.
\begin{center}
\begin{tabular}{rcl}
$d(0)=d(a)$ & = & 0 \\
$d(p_1+p_2)=d(p_1\cdot p_2)$ & = & $\max\{d(p_1),d(p_2)\}$ \\
$d(\binstar{p_1}{p_2})$ & = & $\max\{1+d(p_1),d(p_2)\}$ \\
\end{tabular}
\end{center}
\begin{lemma}
\label{lem:congr_mult}
For all $p,q,r\in T$ such that $\nfmult{p\cdot q}{r}$ and 
$\congr{q}{r}$ at least one of the following always holds:
\begin{enumerate}
\item There exists an $s\in T$ such that $\bisim{p\cdot q\cdot r}{s\cdot r}$
      and $\nfmult{s}{r}$ and $\congr{s}{r}$ and $d(s)\leq d(p\cdot q)$, or
\item There exists an $s\in T$ such that $\bisim{r}{s\cdot 0}$ and 
      $\nf{s\cdot 0}$ and $d(s)\leq 1+d(p\cdot q)$. 
\end{enumerate}
\end{lemma}
\begin{proof}
We first apply (strong) induction towards $d(p)$, thereby generalizing over
all variables, and subsequently induction towards the structure of $p$,
thereby generalizing over $q$ and $r$. Note that the case for $p\equiv 0$ 
can be solved directly by choosing $s\equiv 0$. If $p\equiv a$ for some $a\in A$ 
then we distinguish between two cases: (1) if $\bisim{q\cdot r}{r}$ then set 
$s\equiv a$, (2) otherwise set $s\equiv a\cdot q$. For the situations $p\equiv p_1+p_2$
and $p\equiv p_1\cdot p_2$, we first consider the cases where the $p$-induction
hypothesis for both operands corresponds with the first possibility in this lemma.
\begin{itemize}
\item If $p\equiv p_1+p_2$ then by induction we can derive $s_1,s_2\in T$ such
      that $\bisim{p_1\cdot q\cdot r}{s_1\cdot r}$ and $\bisim{p_2\cdot q\cdot r}{s_2\cdot r}$.
      Choose $s\equiv s_1+s_2$.
\item If $p\equiv p_1\cdot p_2$ then first use induction to derive $\bisim{p_2\cdot q\cdot r}{s_2\cdot r}$.
      Then, apply induction again setting $q\equiv s_2$ to derive $\bisim{p_1\cdot s_2\cdot r}{s\cdot r}$.
\end{itemize}
Note that for both the cases $p\equiv p_1+p_2$ and $p\equiv p_1\cdot p_2$, the
second case in the result is directly satisfied. 

If $p\equiv\binstar{p_1}{p_2}$ then first obtain $s_2\in T$ such
that $\bisim{p_2\cdot q\cdot r}{s_2\cdot r}$ via induction. Now first suppose
that $\bisim{\binstar{p_1}{s_2}\cdot r}{r}$ and observe that in this case we
have $\bisim{\bisim{r}{\binstar{(\binstar{p_1}{s_2})}{0}}}{\binstar{(p_1+s_2)}{0}}$.
Now $s\equiv\binstar{(p_1+s_2)}{0}$ is a witness for the second result. For
the remainder of this case, suppose that $\bisim{\binstar{p_1}{s_2}\cdot r}{r}$
does not hold.

Assume there exists a $p_1\longrightarrow^+t$ such that $\bisim{t\cdot\binstar{p_1}{s_2}\cdot r}{r}$
and observe that now $\bisim{r}{t\cdot\binstar{(p_1+s_2\cdot t)}{0}}$ holds. If there
exists a $s_2\longrightarrow^+u$ such that $\bisim{(u\cdot t)\cdot\binstar{(p_1+s_2\cdot t)}{0}}{\binstar{(p_1+s_2\cdot t)}{0}}$
then we have $\bisim{r}{\binstar{(t\cdot u)}{0}}$. As $d(t)<d(\binstar{p_1}{s_2})$ we
may use induction to obtain an $s\in T$ such that $\bisim{t\cdot u\cdot\binstar{(t\cdot u)}{0}}{s\cdot\binstar{(t\cdot u)}{0}}$,
which leads to a witness for the second result, otherwise, choose $s\equiv\binstar{t\cdot(p_1+s_2\cdot t)}{0}$ for
the second result.

If there does not exists a $p_1\longrightarrow^+t$ such that $\bisim{t\cdot\binstar{p_1}{s_2}\cdot r}{r}$
then $s\equiv\binstar{p_1}{s_2}$ can be chosen as a witness for the first result of the Lemma.
\end{proof}
Lemma \ref{lem:congr_ex} is very similar to Lemma \ref{lem:congr_mult}, but
it is required as a separate result due to the aforementioned absence of a
neutral element under sequential composition.
\begin{lemma}
\label{lem:congr_ex}
For all $p,r\in T$ such that $\nfmult{p}{r}$ 
at least one of the following always holds:
\begin{enumerate}
\item There exists a $q\in T$ such that $\bisim{p\cdot r}{q\cdot r}$
      and $\nfmult{q}{r}$ and $\congr{q}{r}$ and $d(q)\leq d(p)$, or
\item There exists a $q\in T$ such that $\bisim{r}{q\cdot 0}$ and 
      $\nf{q\cdot 0}$ and $d(q)\leq 1+d(p)$. 
\end{enumerate}
\end{lemma}
\begin{proof}
By induction towards the structure of $p$. For the case $p\equiv p_1\cdot p_2$,
one simply invokes Lemma \ref{lem:congr_mult} directly. The cases for 
$p\equiv\binstar{p_1}{p_2}$ are almost the same. Note that induction towards
$d(p)$ is not required in this Lemma as one can use Lemma \ref{lem:congr_mult}
to handle the $s_2\cdot t$-situation for the case $p\equiv\binstar{p_1}{p_2}$. 
\end{proof}
In order to apply the derivation of a bisimilar term satisfying the congruence
property we first formulate a Lemma corresponding to the derivation of $\mathrm{nfmult}$,
followed by a lemma corresponding to the predicate $\mathrm{nf}$.
\begin{lemma}
\label{lem:nf_mult_ex}
For all $p,r\in T$, there exists a $q\in T$ such that
$\bisim{p\cdot r}{q\cdot r}$ and $\nfmult{q}{r}$ and $d(q)\leq d(p)$.
\end{lemma}
\begin{proof}
Apply induction towards the structure of $p$. For the case $p\equiv\binstar{p_1}{p_2}$,
use Lemma \ref{lem:congr_ex} to derive an $s\in T$ such that
$\bisim{s\cdot\binstar{p_1}{p_2}\cdot r}{p_1\cdot\binstar{p_1}{p_2}\cdot r}$,
which results in $\bisim{\binstar{s}{p_2}\cdot r}{\binstar{p_1}{p_2}\cdot r}$,
due to soundness of RSP, and $\congr{s}{\binstar{s}{p_2}\cdot r}$ if the
first result of Lemma \ref{lem:congr_ex} holds. Otherwise, in case of the
second result of Lemma \ref{lem:congr_ex}, $\nfmult{s\cdot 0}{r}$ is satisfied
directly.
\end{proof}
\begin{lemma}
\label{lem:normalization}
For all $p\in T$ there exists a $q\in T$ such that $\bisim{p}{q}$
and $\nf{q}$ and $d(p)\leq d(q)$.
\end{lemma}
\begin{proof}
Apply induction towards the structure of $p$ and handle the case $p\equiv\binstar{p_1}{p_2}$
as indicated in Lemma \ref{lem:nf_mult_ex}.
\end{proof}
\section{Completeness}
\label{sec:completeness}
We will use the normal form obtained under bisimilarity in section \ref{sec:normalization}
to finish the completeness proof. This requires some administrative work but is not very deep. 
Clearly the implication $\bisim{p}{q}\Rightarrow p=q$ is our proof obligation. The
completeness proof uses induction towards $d(p)$. Lemma \ref{lem:normalization} can
be used to derive an $r\in T$ such that $\bisim{p}{\bisim{r}{q}}$ and such that
$\nf{r}$ holds. By transitivity, these two equalities may be solved separately. This
mainly comes down to two steps where deriving axiomatic equality is mostly done by
invocations of Lemma \ref{lem:next}.
\begin{enumerate}
\item Using Lemma \ref{lem:next} and induction to reduce to an equality of a form
      similar to $t\cdot\binstar{p}{q}=\binstar{r}{s}$ for $p\longrightarrow^+t$ or 
      $\binstar{p}{q}=\binstar{r}{s}$, which in turn can be reduced to an equality
      of the form $x\cdot\binstar{p}{q}=y\cdot\binstar{p}{q}$ for $p\longrightarrow^+x$
      and $p\longrightarrow^+y$ via the axiom RSP and Lemma \ref{lem:next}.
\item Using $\nf{\binstar{p}{q}}$ and application of Lemma \ref{lem:congruence} 
      to solve $x=y$ using the $d$-induction hypothesis for completeness with
      regard to $d(\binstar{p}{q})$.
\end{enumerate}
We require the definition of two more predicates to aid in compact formulation of
the succeeding lemmas. We say that a set $M\subseteq A\times V$ is
a tail of $p$ (notation: $\tail{p}{M}$) if there exists a $q\in T$ such
that $p\longrightarrow^*q$ and for all $(a,u)\in M$ it holds that $q\step{a}u$. In
addition, we say that a term is next-provable (notation: $\nextt{p}$) if for all
$v\in V$ and for all $p\step{a}u$ for $u\in V$ it holds that
$\bisim{u}{v}$ implies $\teq{u}{v}$. Using the $d$-induction hypothesis, we
can prove the following lemma.
\begin{lemma}
\label{lem:core}
For $M,N,K,L\subseteq A\times V$ and $p,q\in T$ such
that $\tail{p}{M}$, $\tail{p}{K}$, $\nextt{\sigma(N)}$, $\nextt{\sigma(L)}$ and
$\nf{\binstar{p}{q}}$ it holds that 
\begin{center}
\begin{math}
\bisim{\sigma(M)\cdot\binstar{p}{q}+\sigma(N)}{\sigma(K)\cdot\binstar{p}{q}+\sigma(L)}
\end{math}
\end{center}
implies 
\begin{center}
\begin{math}
\sigma(M)\cdot\binstar{p}{q}+\sigma(N)=\sigma(K)\cdot\binstar{p}{q}+\sigma(L)
\end{math}
\end{center}
\end{lemma}
\begin{proof}
Apply Lemma \ref{lem:next} and note that the premises arising from steps
from $\sigma(N)$ or $\sigma(L)$ are directly resolved. Furthermore, apply
Lemma \ref{lem:congruence} and the induction hypothesis for every step arising 
from $\sigma(M)$ and $\sigma(K)$.
\end{proof}
Lemma \ref{lem:complete_split} is crucial in the transformation of the
equality towards application of Lemma \ref{lem:core}. We define the
predicate $\mathrm{obl}$ to formalize the proof obligation as follows. Say
that $\obl{p}{q}{r}$ holds true if and only if for all $M,N\subseteq
A\times V$ such that:
\begin{enumerate}
\item $\bisim{\sigma(M)\cdot\binstar{p}{q} + \sigma(N)}{r}$, 
\item $\tail{p}{M}$, 
\item $\nextt{\sigma(N)}$, 
\item $\nextt{q}$, 
\item $\nf{\binstar{p}{q}}$ and
\item for all $x,y\in T$ such that $d(x)<d(\binstar{p}{q})$
      we have $\bisim{x}{y}$ implies $x=y$,
\end{enumerate}
the conclusion $\sigma(M)\cdot\binstar{p}{q}+\sigma(N)=r$ follows.
\begin{lemma}
\label{lem:complete_split}
For all $p,q,r,s\in T$ it holds that $\obl{p}{q}{s}$ implies
$\obl{p}{q}{r\cdot s}$.
\end{lemma}
\begin{proof}
Apply induction towards the structure of $r$, thereby generalizing over
all other variables. For $r\equiv 0$, use Lemma \ref{lem:next}. For
$r\equiv a$ for some $a\in A$, use Lemma \ref{lem:next} and
the premise $\obl{p}{q}{s}$. For $r\equiv r_1+r_2$, one must first derive
sets $M_1,M_2,N_1,N_2\subseteq A\times V$ such that:
$M=M_1\cup M_2$ and $N=N_1\cup N_2$ and
\begin{center}
\begin{math}
\bisim{\sigma(M_1)\cdot\binstar{p}{q}+\sigma(N_1)}{r_1\cdot s}\quad\textrm{and}\quad
\bisim{\sigma(M_2)\cdot\binstar{p}{q}+\sigma(N_2)}{r_2\cdot s}.
\end{math}
\end{center}
The result then follows from the respective induction hypotheses for
$r_1$ and $r_2$. Note that the case $r\equiv r_1\cdot r_2$ is immediate
due to the setup of this lemma. 

The remaining case is $r\equiv\binstar{r_1}{r_2}$. Apply the axiom RSP and note 
that the reversal of this axiom is sound, this leads to the following proof obligation
\begin{center}
\begin{math}
\sigma(M)\cdot\binstar{p}{q}+\sigma(N)=r_1\cdot(\sigma(M)\cdot\binstar{p}{q}+\sigma(N))+r_2\cdot s
\end{math}
\end{center}
One first derives four sets $M_1,M_2,N_1,N_2\subseteq
A\times V$ such that
\begin{center}
\begin{math}
\bisim{\sigma(M_1)\cdot\binstar{p}{q}+\sigma(N_1)}{r_1\cdot(\sigma(M)\cdot\binstar{p}{q}+\sigma(N)}
\end{math}
\end{center}
and
\begin{center}
\begin{math}
\bisim{\sigma(M_2)\cdot\binstar{p}{q}+\sigma(N_2)}{r_2\cdot s}.
\end{math}
\end{center}
These two equalities can be resolved via induction. The premise
$\obl{p}{q}{\sigma(M)\cdot\binstar{p}{q}+\sigma(N)}$ is a result
of Lemma \ref{lem:core}.
\end{proof}
The following lemma is the analog of Lemma \ref{lem:complete_split}
and only required due to the aforementioned absence of a neutral 
element under multiplication.
\begin{lemma}
\label{lem:complete}
For $p,q,r\in T$ and $M,N\subseteq A\times V$
such that $\bisim{\sigma(M)\cdot\binstar{p}{q}+\sigma(N)}{r}$, 
$\tail{p}{M}$, $\nextt{\sigma(N)}$, $\nextt{q}$ and $\nf{\binstar{p}{q}}$
we have $\sigma(M)\cdot\binstar{p}{q}+\sigma(N)=r$, given the $d$-induction
hypothesis with regard to $d(\binstar{p}{q})$. 
\end{lemma}
\begin{proof}
By induction towards the structure of $r$, thereby generalizing over
all other variables. One must use Lemma \ref{lem:complete_split} here
for the case $r\equiv r_1\cdot r_2$.
\end{proof}
Before the completeness proof can be finished, we require two lemmas to
work towards a star-reduct on one side of the equality. Once we have an
equality in this form, Lemma \ref{lem:summation} can be applied to invoke
Lemma \ref{lem:complete}.
\begin{lemma}
\label{lem:towards_split}
For all $p,q,r\in T$ such that $\bisim{p\cdot r}{q}$ and
$\nf{p\cdot r}$ and $\nextt{r}$ we have $p\cdot r=q$, under the
$d$-induction hypothesis with regard to $d(p\cdot r)$.
\end{lemma}
\begin{proof}
The most straightforward way to prove this lemma is a setup close
to Lemma \ref{lem:next}. For this purpose, define the helper
function $\mathrm{mult}:V\times T\rightarrow T$ as follows: $\mult{\surd}{q}=q$ and $\mult{p}{q}=
p\cdot q$ for $p\in T$. One may then prove: for all $p\step{a}u$ and $v\in V$:
$\bisim{\mult{u}{r}}{v}$ implies $\teq{\mult{u}{r}}{v}$ by induction
towards the structure of $p$. For the case $p\equiv\binstar{p_1}{p_2}$,
use Lemma \ref{lem:summation} to prepare in a suitable form for
application of Lemma \ref{lem:complete}.
\end{proof}
\begin{lemma}
\label{lem:towards}
For all $p,q\in T$ such that $\bisim{p}{q}$ and $\nf{p}$
it holds that $p=q$ given the induction hypothesis towards $d(p)$.
\end{lemma}
\begin{proof}
Analogous to Lemma \ref{lem:towards_split}.
\end{proof}
We may now prove the completeness theorem.
\begin{theorem}
\label{thm:completeness}
For all $p,q\in T$ such that $\bisim{p}{q}$ it holds that $p=q$.
\end{theorem}
\begin{proof}
Apply strong induction towards $d(p)$, thereby generalizing over $p$
and $q$. By Lemma \ref{lem:normalization} there exists an $r\in T$ 
such that $\bisim{p}{r}$ and $\nf{r}$ and $d(r)\leq d(p)$.
Lemma \ref{lem:towards} can now be applied to solve $r=p$ and $r=q$
separately.
\end{proof}
\section{Formalization}
\label{sec:formalization}
All the results in this paper have been formalized and verified using
version 8.4pl4 of the Coq proof-assistant (\cite{coq97}). The Coq code
is available at
\begin{center}
\texttt{https://github.com/allanvanhulst/mscs/}
\end{center}


We import Coq-libraries for lists, maximality natural numbers, and for Presburger arithmetic.
\begin{verbatim}
Require Import List.
Require Import Max.
Require Import Omega.
\end{verbatim}
We then encode the law of the excluded middle (LEM) as an additional
axiom. The addition of this axiom is consistent with Coq (cf. \cite{coq97}).
\begin{verbatim}
Axiom LEM : forall (P : Prop), P \/ ~ P.
\end{verbatim}
The set of actions $A$ is a parameter for the theory $T$.
\begin{verbatim}
Parameter A : Set.

Inductive T :=
  | zero : T
  | act : A -> T
  | plus : T -> T -> T
  | mult : T -> T -> T
  | star : T -> T -> T.
\end{verbatim}
Notations for 0, plus, multiplication and the binary Kleene-star are
subsequently introduced. The Unicode character $\cdot$ is used for 
sequential composition. Note that non-ASCII characters are valid
within the Coq proof assistant but may render in different ways on
different platforms. Moreover, the operators for plus and Kleene-star 
are overloaded and assigned a new meaning with regard to $T$ here. This 
is perfectly valid in Coq but due to parsing limitations, one cannot 
re-assign associativity or precedence to an already existing operator. 
However, this is not a problem here.
\begin{verbatim}
Notation "0" := zero.
Notation "p + q" := (plus p q).
Notation "p · q" := (mult p q) (at level 45, right associativity). 
Notation "p * q" := (star p q).
\end{verbatim}
We then define the set $V$ as an inductive construction
for either $\surd$ (i.e. \texttt{term}) or the constructor 
$\texttt{emb}:T\rightarrow V$. 
\begin{verbatim}
Inductive V :=
  | term : V
  | emb : T -> V.
\end{verbatim}
The transition relation $\longrightarrow\subseteq V\times
A\times V$ is then defined as a ternary predicate
having a corresponding notation $\texttt{p -(a)-> q}$.
\begin{verbatim}
Inductive step : V -> A -> V -> Prop := 
  | step_act : forall (a : A), step (emb (act a)) a term
  | step_plus_left : forall (p q : T) (v : V) (a : A), 
    step (emb p) a v -> step (emb (p + q)) a v
  | step_plus_right : forall (p q : T) (v : V) (a : A),
    step (emb q) a v -> step (emb (p + q)) a v
  | step_mult_left : forall (p q p' : T) (a : A),
    step (emb p) a (emb p') -> step (emb (p · q)) a (emb (p' · q))
  | step_mult_right : forall (p q : T) (a : A),
    step (emb p) a term -> step (emb (p · q)) a (emb q)
  | step_star_left : forall (p q p' : T) (a : A),
    step (emb p) a (emb p') -> step (emb (p * q)) a (emb (p' · (p * q)))
  | step_star_term : forall (p q : T) (a : A),
    step (emb p) a term -> step (emb (p * q)) a (emb (p * q))
  | step_star_right : forall (p q : T) (a : A) (v : V),
    step (emb q) a v -> step (emb (p * q)) a v.

Notation "p '-(' a ')->' q" := (step p a q) (at level 30).
\end{verbatim}
Bisimilarity then directly follows the definition in this paper.
\begin{verbatim}
Definition bisim (u v : V) : Prop := exists (R : V -> V -> Prop),
  R u v /\ forall (x y : V), R x y -> 
    (x = term <-> y = term) /\
    (forall (a : A) (x' : V), x -(a)-> x' ->
      exists (y' : V), y -(a)-> y' /\ R x' y') /\
    (forall (a : A) (y' : V), y -(a)-> y' ->
      exists (x' : V), x -(a)-> x' /\ R x' y').
\end{verbatim}
The simplest way to encode the axioms in Coq while having the most
possible certainty that no interference with the internal Coq-axioms
takes place is to inductively define axiomatic equality as a binary
predicate. This makes the proof somewhat cumbersome. For instance, 
every instance where transitivity or symmetry of the axiom system is 
applied needs to be specified explicitly. Essentially, we are encoding
the derivation of axiomatic equality as a proof-tree here.
\begin{verbatim}
Inductive ax : T -> T -> Prop :=
  | refl : forall (x : T), ax x x
  | symm : forall (x y : T), ax x y -> ax y x
  | trans : forall (x y z : T), ax x y -> ax y z -> ax x z
  | comp_plus : forall (w x y z : T), ax w y -> ax x z -> ax (w + x) (y + z)
  | comp_mult : forall (w x y z : T), ax w y -> ax x z -> ax (w · x) (y · z)
  | comp_star : forall (w x y z : T), ax w y -> ax x z -> ax (w * x) (y * z)
  | B1 : forall (x y : T), ax (x + y) (y + x)
  | B2 : forall (x y z : T), ax ((x + y) + z) (x + (y + z))
  | B3 : forall (x : T), ax (x + x) x
  | B4 : forall (x y z : T), ax ((x + y) · z) (x · z + y · z)
  | B5 : forall (x y z : T), ax ((x · y) · z) (x · (y · z))
  | B6 : forall (x : T), ax (x + 0) x
  | B7 : forall (x : T), ax (0 · x) 0
  | BKS1 : forall (x y : T), ax (x · (x * y) + y) (x * y)
  | BKS2 : forall (x y z : T), ax ((x * y) · z) (x * (y · z))
  | RSP : forall (x y z : T), ax x (y · x + z) -> ax x (y * z).
\end{verbatim}
The notation \texttt{<=>} is used for bisimilarity, and the 
notation  \texttt{==}  is used for axiomatic equality.
\begin{verbatim}
Notation "u '<=>' v" := (bisim u v) (at level 25).
Notation "p '==' q"  := (ax p q) (at level 25).
\end{verbatim}
This is all the Coq-code required to express the completeness theorem
at the very end of the Coq-proof.
\begin{verbatim}
Theorem completeness : forall (p q : T), emb p <=> emb q -> p == q.
\end{verbatim}
The Coq-proof in its present form consists of 3959 lines of code, 
and is divided into 73 lemmas and one theorem. Note that most of the
lemmas are very simple. An example is shown below.
\begin{verbatim}
Lemma step_plus_fmt : forall (p q : T) (a : A) (u : V),
  emb (p + q) -(a)-> u -> emb p -(a)-> u \/ emb q -(a)-> u.
Proof.
  intros p q a u H ; inversion H ; auto.
Qed.
\end{verbatim}
\section{Conclusions}
\label{sec:conclusions}
Compared to the proof method proposed in \cite{gf20}, the
completeness result presented in this work is simpler, although a rather
technical treatment was essentially unavoidable. The formalization in the
Coq proof assistant ensures guaranteed correctness of the result. The method
is not very distant to an algorithmic rewriting procedure and an actual
implementation may be constructed by reasonable effort.

The theory and axiom system originally considered by Milner \cite{mil84} allows for a
congruence property to be formulated in a very similar way as was done here.
Again, we require that for all $p\longrightarrow^+t$ and $p\longrightarrow^+u$
we have $\bisim{t\cdot p^*\cdot q}{u\cdot p^*\cdot q}$ implies $\bisim{t}{u}$.
This follows if for all $p\longrightarrow^+t$ the following two conditions are
satisfied
\begin{enumerate}
\item if $\bisim{t\cdot p^*\cdot q}{(t+1)\cdot p^*\cdot q}$ then $t\downarrow$ and
\item if $t\downarrow$ and $t\step{a}t'$ then there does not exist a step
      $p^*\cdot q\step{a}r$ such that $\bisim{t'\cdot p^*\cdot q}{r}$.
\end{enumerate}
It is not unlikely that the difficulty in obtaining such normal forms lies
almost entirely in ensuring that the first of these two conditions is met.
Assume that the predicate $\mathrm{nfmult}$ now expresses this new congruence
property recursively and further assume that the theory $T$ now has the unary 
Kleene star and 1. Resolving the following conjecture may be the key to finding 
a solution to Milner's question.
\begin{conjecture}
\label{con:milner}
For $p\in T$ there exists a $q\in T$ such that 
$\bisim{p}{q}$ and $\nfmult{q}{1}$ and $d(q)\leq d(p)$.
\end{conjecture}
At this point, it cannot be excluded that Conjecture \ref{con:milner} can
be resolved using a proof method similar to the one applied in this paper.

\end{document}